\documentclass{agt}

\usepackage{amsfonts}
\usepackage{amsmath}
\usepackage{stmaryrd}
\usepackage{amssymb}
\usepackage{graphicx}
\usepackage{appendix}
\usepackage{subfigure}
\usepackage{url}
\usepackage{hyperref}
\usepackage{enumerate}

\begin{document}
\title{Computational Topology for Approximations of Knots}{Computational Topology for Approximations of Knots}
\author{J. Li \and T. J. Peters \and K. E. Jordan}{}

\address[ji.li@uconn.edu]{J. Li}{Department of Mathematics, University of Connecticut, Storrs, CT, USA.}
\address[tpeters@cse.uconn.edu]{T. J. Peters}{Department of Computer Science and Engineering,
        University of Connecticut, Storrs, CT, USA.} 
\address[kjordan@us.ibm.com]{K. E. Jordan}{IBM T.J. Watson Research, Cambridge Research Center, Cambridge, MA, USA.}

\date{\today}

\begin{abstract} 
The preservation of ambient isotopic equivalence under piecewise linear (PL) approximation for smooth knots are prominent in molecular modeling and simulation. Sufficient conditions are given regarding:

\begin{enumerate}
\item Hausdorff distance, and
\item a sum of total curvature and derivative. 
\end{enumerate}

High degree B\'ezier curves are often used as smooth representations, where computational efficiency is a practical concern. Subdivision can produce PL approximations for a given B\'ezier curve, fulfilling the above two conditions. The primary contributions are:
\begin{enumerate}[(i)]
\item \emph{a priori} bounds on the number of subdivision iterations sufficient to achieve a PL approximation that is ambient isotopic to the original B\'ezier curve, and
\item improved iteration bounds over those previously established.
\end{enumerate}
\end{abstract}

\section{Introduction}
\label{sec:intro}

\begin{figure}[h!]
\centering
    \subfigure[Unknot VS. Knot]
    {
   \includegraphics[height=2.7cm]{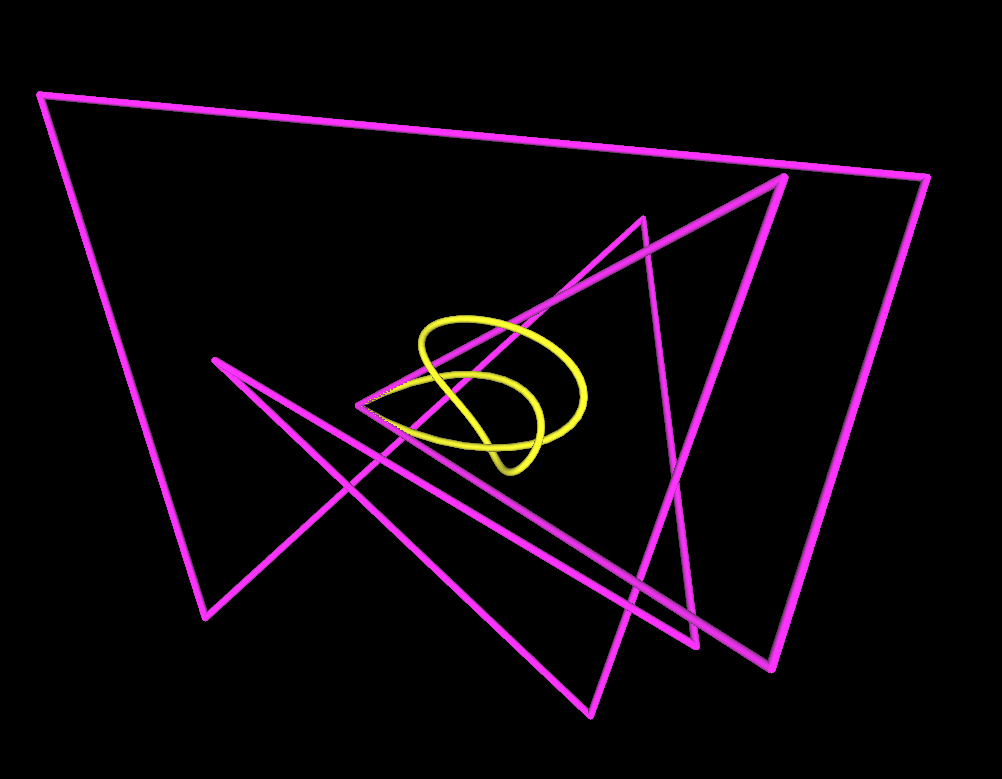}
    \label{fig:ku0}
    }
    \subfigure[An intermediate step]
    {
   \includegraphics[height=2.7cm]{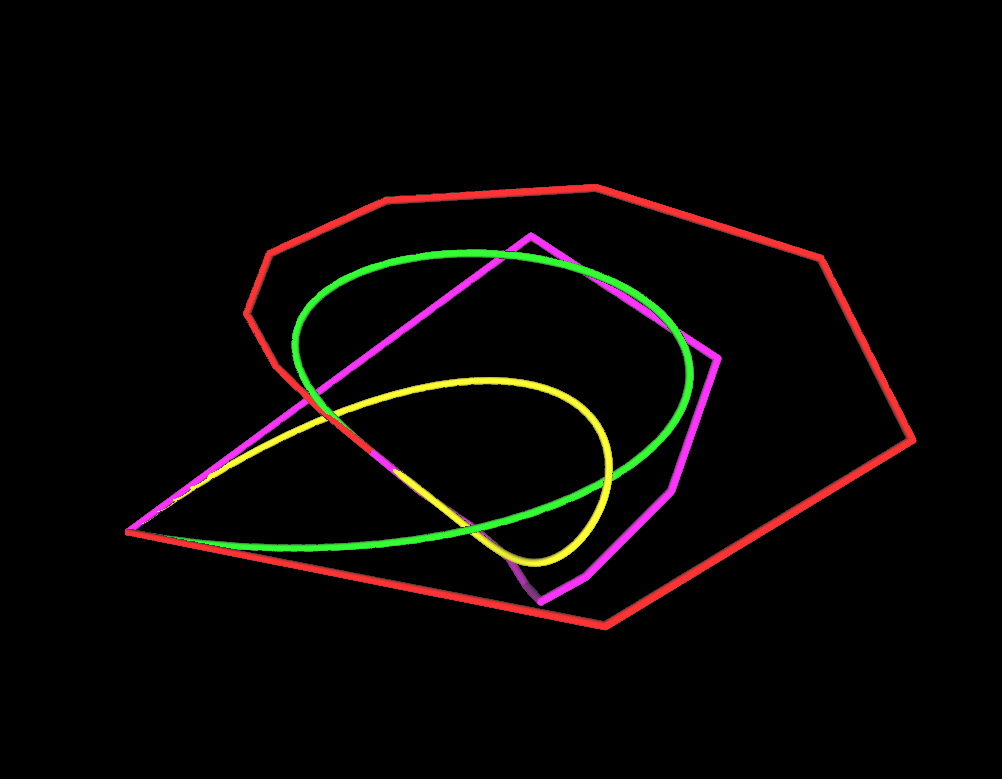}
    \label{fig:ku1}
    }
    \subfigure[Knot VS. Knot]
    {
   \includegraphics[height=2.7cm]{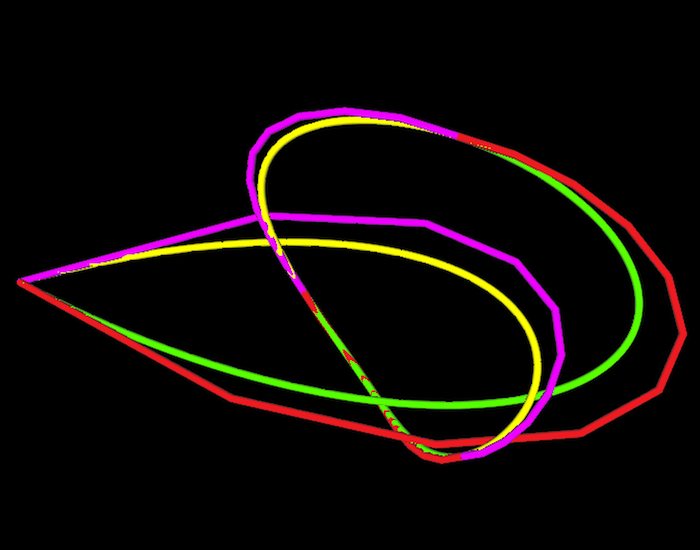}
    \label{fig:ku2}
    }
    \caption{Ambient isotopic approximation}
    \label{fig:aa}
\end{figure}

Figure~\ref{fig:ku0} demonstrates an example of topological difference, where a knotted B\'ezier curve is defined by an unknotted control polygon \cite{JL2012}. Subdivision is then used to generate new control polygons. Figure~\ref{fig:ku1} shows the control polygon after one subdivision, where the topological difference remains. Figure~\ref{fig:ku2} shows the control polygon after two subdivisions, where the control polygon obtains the same topology as the underlying curve. 

The images are illustrative and a curve visualization tool \cite{TJPweb} was used to experimentally create these examples.  Rigorous proofs of the topological difference between the B\'ezier curve and its initial control polygon were formulated \cite[Section 2]{JL2012}. This serves as a cautionary note that graphics used to approximate a curve may not have isotopic equivalence. Additional rigorous topological analysis is important, as described here. Figure~\ref{fig:ku1} and~\ref{fig:ku2} are visual examples that show successive subdivisions eventually produce topologically correct PL approximations. The advantage of the bounds given here are discussed in Remark~\ref{re:better}.

\subsection{Topological background}
\label{ssec:topback}

There is contemporary interest \cite{Amenta2003, L.-E.Andersson2000, ChoMaekawa1996, Lance2009, Moore_Peters_Roulier2007} to preserve topological characteristics such as homeomorphism and ambient isotopy between an initial geometric model and its approximation. Ambient isotopy is a continuous family of homeomorphisms $H: X \times [0,1] \rightarrow Y$ such that 
$$H(X, 0)=X\ \textup{and}\ H(X, 1)=Y,$$ 
for topological spaces $X$ and $Y$ \cite{Hirsch}. It is particularly applicable for time varying models, such as the writhing of molecules. 

A B\'ezier curve is characterized by an indexed set of  points, which forms a piecewise linear ($PL$) approximation of the curve, called a control polygon. The de Casteljau algorithm \cite{G.Farin1990} is a subdivision algorithm associated to B\'ezier curves which recursively generates control polygons more closely approximating the curve under Hausdorff distance \cite{J.Munkres1999, Nairn-Peters-Lutterkort1999}. 

An earlier algorithm \cite{TJP2011} establishes an isotopic approximation over a broad class of parametric geometry, but can not provide the number of subdivision iterations for B\'ezier curves. Other recent papers \cite{Burr2012, LineYap2011} present algorithms to compute isotopic PL approximation for $2D$ algebraic curves.  Computational techniques for establishing isotopy and homotopy have been established regarding algorithms for point-cloud  by ``distance-like functions'' \cite{Chazal2005}. Ambient isotopy under subdivision was previously established \cite{Moore_Peters_Roulier2007} for $3D$ B\'ezier curves of low degree (less than 4).

Recent progress regarding isotopy under certain convergence criteria has been made \cite{DenneSullivan2008, JL-isoconvthm, bez-iso}. In particular, Denne and Sullivan proved that for homeomorphic curves, if their distance and angles between the first derivatives are within some given bounds, then these curves are ambient isotopic \cite{DenneSullivan2008}. This result has been applied to B\'ezier curves \cite{bez-iso}. Here we present an alternative set of conditions for ambient isotopy that is explicitly constructed. It is useful for applications that require explicit maps between initial and terminal configurations. Remark~\ref{rmk:rbd} will show that there is no need to test first derivatives. Instead, we test global conditions of distance and total curvature. It may also be useful when the conditions here are easier to be verified than those in the previously established method. Furthermore, the subdivision iteration bound established here is an improvement over the previous one (Remark~\ref{re:better}).

Moreover, this is alternative to a result regarding existence of ambient isotopy for B\'ezier curves \cite{JiLi}. The pure existence proof requires the convex hulls of sub-control polygons to be contained in a tubular neighborhood determined by a pipe surface and may need more subdivision iterations and produce too many $PL$ segments. The work here removes this convex hull constraint and produces the isotopy using fewer subdivision iterations.

A technique we will use is called {\em pipe surface} \cite{Maekawa_Patrikalakis_Sakkalis_Yu1998}. A pipe surface of radius $r$ of a curve $c(t)$, where $t \in [0,1]$ is given by
\[ \textbf{ p}(t,\theta) = c(t) + r[cos(\theta) \textbf{ n}(t) + 
sin(\theta) \textbf{ b}(t) ], \]
where $\theta \in [0,2\pi]$ and $\textbf{n}(t)$ and $\textbf{b}(t)$ are, respectively, the normal and bi-normal vectors at the point $c(t)$, 
as given by the Frenet-Serret trihedron.

\begin{rem}
The paper \cite{Maekawa_Patrikalakis_Sakkalis_Yu1998} provides the computation of the radius $r$ only for rational spline curves. However, the method of computing $r$ is similar for other  compact, regular, $C^2$, and simple curves, that is, taking the minimum of $1/\kappa_{max}$, $d_{min}$, and $r_{end}$, where $\kappa_{max}$ is the maximum of the curvatures, $d_{min}$ is the minimum separation distance, and $r_{end}$ is the maximal radius around the end points that does not yield self-intersections.
\end{rem}

Pipe surfaces have been studied since the 19th century \cite{Monge}, but the presentation here follows a contemporary source~\cite{Maekawa_Patrikalakis_Sakkalis_Yu1998}. These authors perform a thorough analysis and description of the end conditions of open spline curves.  The junction points of a B\'ezier curve are merely a special case of that analysis. 

We shall state the conditions. We assume throughout this paper that the space curves are parametric, compact, simple (non-self-intersecting) and regular (The first derivatives never vanish). Given two curves, $PL$ and smooth respectively (Usually, the $PL$ curve is an approximation of the smooth curve.), suppose that they are divided into sub-curves. Let $L(t) : [0,1] \rightarrow \mathbb{R}^3$ and  $C(t) : [0,1] \rightarrow \mathbb{R}^3$ be the corresponding $PL$ and smooth sub-curves. We require that $L(0)=C(0)$ and $L(1)=C(1)$. In particular, for a B\'ezier curve, subdivision produces sub-control polygons and the corresponding smooth sub-curves such that each pair of end points between the $PL$ and smooth sub-curves are connected. 

There exists a nonsingular pipe surface of radius $r$ for $C$ \cite{Maekawa_Patrikalakis_Sakkalis_Yu1998}. Denote the disc of radius $r$ centered at $C(t)$ and normal to $C$ as $D_r(t)$. Let a {\em pipe section} to be $\Gamma=\bigcup_{t\in[0,1]}D_r(t)$. Denote the interior as int$(\Gamma)$, and the boundary as $\partial \Gamma$. Note that the boundary $\partial \Gamma$ consists of the nonsingular pipe surface and the end discs $D_r(0)$ and $D_r(1)$. Define $\theta(t): [0,1] \rightarrow [0,\pi]$ by 
$$\theta(t) = \eta(C'(t),L'(t)),$$
where the function $\eta(\cdot,\cdot)$ denotes the angle between two vectors \cite{bez-iso}.

\subsection{Our two conditions}
\label{ssec:our2}

The two primary conditions for this paper are now stated.

\vspace{2ex}
\noindent \textbf {\textit{Conditions 1 \textup{and} 2} } for ambient isotopy are:
\begin{enumerate}
\item $L\setminus \{L(0), L(1)\} \subset \textup{int}(\Gamma)$; and
\item $T_{\kappa}(L) + \max_{t\in[0,1]} \theta(t)< \frac{\pi}{2}$,
\end{enumerate}
where $\Gamma$ is the pipe section of $C$ and $T_{\kappa}(L)$ denotes the total curvature of $L$, i. e. the sum of exterior angles \cite{bez-iso}.

\textit{Conditions 1 \textup{and} 2} will guarantee ambient isotopy between not only the sub-curves $L$ and $C$, but also the whole curves, which is more important.   

\begin{rem}\label{rmk:rbd}
We shall show later that, for a B\'ezier curve, the number of subdivisions for \textit{Condition 2} is at most one more than that for a weaker condition $T_{\kappa}(L) <\frac{\pi}{2}$ (Lemma~\ref{lem:c2wf} in Section~\ref{ssec:nchns}). This allows us to easily remove the burden of testifying the derivatives in order to find $\theta(t)$. 
\end{rem}

\section{Construction of Homeomorphisms}\label{ssec:dah}

Constructing the ambient isotopy here relies upon explicitly constructing a homeomorphism. The explicit construction provides more algorithmic efficiency than only showing the existence of  these equivalence relations.

\begin{lem}\label{lem:hp12}
Suppose $L$ is a sub-control polygon and $C$ is the corresponding B\'ezier sub-curve. Then \textit{Conditions 1 \textup{and} 2} can be achieved by subdivision.
\end{lem}
\begin{proof}
By the convergence in Hausdorff distance under subdivision, sufficiently many subdivision iterations will produce a control polygon that fits inside a nonsingular pipe surface. Furthermore, by the Angular Convergence \cite[Theorem 4.1]{bez-iso} and the lemma \cite[Lemma 5.3]{bez-iso}, possibly more subdivisions will ensure that each sub-control polygon lies in the corresponding nonsingular pipe section, which is the \textit{Condition 1}. Denote the number of subdivision iterations to achieve this by $\iota_1$. 

By the Angular Convergence, $T_{\kappa}(L)$ converges to $0$ under subdivision. Because the discrete derivative of the control polygon converges to the derivative of the B\'ezier curve \cite{Morin_Goldman2001} under subdivision, $\theta(t)$ converges to $0$ for each $t \in [0,1]$. So \textit{Condition 2} will be achieved by sufficiently many subdivision iterations, say $\iota_2$. (The Details to find $\iota_1$ and $\iota_2$ are in Section~\ref{ssec:nchns}.) 
\end{proof}

\begin{rem}
To obtain some intuition for these conditions, restrict our attention to a B\'ezier curve. Consider $L$ to be a sub-control polygon and $C$ to be the corresponding sub-curve. \textit{Condition 1} will ensure that $L$ lies inside a nonsingular pipe section, while \textit{Condition 2} will ensure a local homeomorphism between $L$ and $C$. In particular, \textit{Conditions 1 \textup{and} 2} will be sufficient for us to establish the one-to-one correspondence using normal discs of $C$. 
\end{rem}

\textbf{\textit{Conditions 1 \textup{and} 2} are assumed in the rest of the section.}

Define a function $\tilde{L}(t): [0,1] \rightarrow L$ by letting
\begin{align}\label{eq:tl} \tilde{L}(t) = D_r(t) \cap L, \end{align}
where $D_r(t)$ is the normal disc of $C$ at $t$.

Define a map $h: C \rightarrow L$ for each $p \in C$ by setting
\begin{align} \label{eq:hh}  h(p) = \tilde{L}(C^{-1}(p)) . \end{align} 

We shall show that $h$ is a homeomorphism. The subtlety here is to demonstrate the one-to-one correspondence by showing each normal disc of $C$ intersects $L$ at a single point (which will be the main goal of the following), and intersects $C$ at a single point (which will be easy), under the assumption of \textit{Conditions 1 \textup{and} 2}. 

\subsection{Outline of the proof}\label{sapp:ideas}
\begin{figure}[h!]
\centering
\includegraphics[height=5cm]{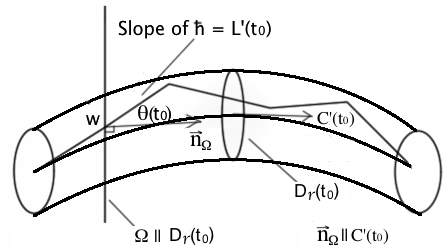}
\caption{Each normal disc intersects $L$ at a single point}
\label{fig:1-1}
\end{figure}

For an arbitrary $t_0 \in [0,1]$, the associated normal disc is denoted as $D_r(t_0)$. Following is the sketch of proving that $D_r(t_0)$ intersects $L$ at a single point. (See Figure~\ref{fig:1-1}.)

\begin{enumerate}
\item \textit{The essential initial steps are to select a non-vertex point of $L$, denoted as $w$, a plane, denoted as $\Omega$, and an angle, denoted as $\theta(t_0)$:}
\begin{enumerate}
\item Define $w$ and $\Omega$: Pick a line segment of $L$ whose slope is equal to $L'(t_0)$, denoted as $\hbar$. Choose an interior point of $\hbar$, denoted as $w$. Let $\Omega$ be the plane that contains $w$ and is parallel to $D_r(t_0)$. (We use $w$ to define two sub-curves of $L$, a `left' sub-curve which terminates at $w$, denoted as $L_l$, and a `right' sub-curve which begins at $w$, denoted as $L_r$.)

\item Consider $\eta(C'(t_0), L'(t_0))=\theta(t_0)$.  Since $\Omega$ is parallel to $D_r(t_0)$, a normal vector of $\Omega$, denoted by $\vec{n}_\Omega$ has the same direction as $C'(t_0)$ and  $\eta(\vec{n}_{\Omega},\hbar)=\eta(C'(t_0), L'(t_0))=\theta(t_0)$. 
\end{enumerate}

\begin{figure}[h!]
\begin{center}
\includegraphics[height=5cm]{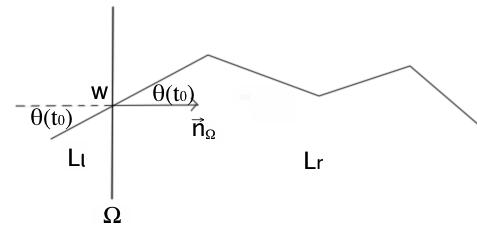}
\end{center}
\caption{Similar angles $\theta(t_0)$}
\label{fig:sim}
\end{figure}

\textit{Remark:} {\small Since $\eta(\vec{n}_{\Omega},\hbar)=\theta(t_0)$, \textit{Condition 2} implies that $T_{\kappa}(L)+\eta(\vec{n}_{\Omega},\hbar)=T_{\kappa}(L)+\theta(t_0)<\frac{\pi}{2}.$ Since $w$ is an interior point of $\hbar$, the angle determined by $\Omega$ and $L_l$, and the angle determined by $\Omega$ and $L_r$, have the same measure $\theta(t_0)$, as shown in Figure~\ref{fig:sim}. So we obtain the similar inequalities $T_{\kappa}(L_l)+\theta(t_0)<\frac{\pi}{2}$ and $T_{\kappa}(L_r)+\theta(t_0)<\frac{\pi}{2}$, which will be crucial.}

\item Prove, by \textit{Condition 2}, that $\Omega  \cap  L_r=w$ . Similarly, show that $\Omega  \cap  L_l=w$. So $\Omega \cap L=w$. (Lemma~\ref{lem:tlpt})

\item Prove that any plane parallel to $\Omega$ intersects $L$ at no more than a single point. (Lemma~\ref{lem:prol})

\item Since $D_r(t_0) \parallel \Omega$, it will follow that $D_r(t_0)$ intersects $L$ no more than a single point. Show, using \textit{Condition 1}, that $D_r(t_0)$ must intersect $L$, and hence $D_r(t_0) \cap L$ is a single point. (Lemma~\ref{lem:dl1})

\end{enumerate}

\subsection{Preliminary lemmas for homeomorphisms}\label{app:pflem}

In order to work with total curvatures of $PL$ curves, an extension of the spherical triangle inequality \cite{GeoTop2005}, given in Lemma~\ref{lem:bt}, will be useful, similar to previous usage by Milnor \cite{Milnor1950}. 

\begin{figure}[h!]
\begin{center}
\includegraphics[height=6cm]{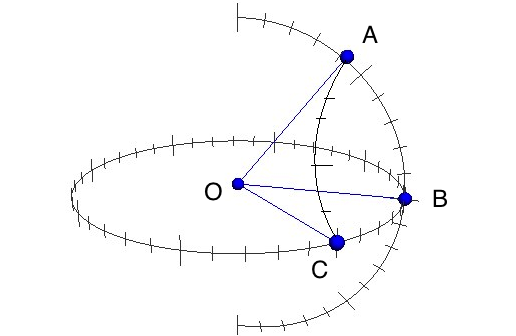}
\end{center}
\caption{Spherical triangle $\triangle{ABC}$}
\label{fig:sph}
\end{figure}

\textbf{ Spherical triangle inequalities:} Consider Figure~\ref{fig:sph}, and the three angles $\angle{AOB}$,  $\angle{BOC}$, and  $\angle{AOC}$, formed by three unit vectors $\overrightarrow{OA}$, $\overrightarrow{OB}$, and $\overrightarrow{OC}$. (Note the common end point $O$. When we consider angles between vectors that do not share such a common end point, we move the vectors to form a common end point.) Denote the arc length of the curve from $A$ to $B$ as $\ell(\widehat{AB})$, and similarly for that from $B$ to $C$ as $\ell(\widehat{BC})$ and that from $A$ to $C$ as $\ell(\widehat{AC})$. The triangle inequality, $\ell(\widehat{AB}) \leq \ell(\widehat{BC})+\ell(\widehat{AC})$, of the spherical triangle $\triangle{ABC}$ provides that 
\begin{align}\label{eq:st1} \angle{AOB} \leq \angle{BOC} + \angle{AOC}.\end{align}

\begin{lem}\label{lem:bt}
Suppose that $\vec{v}_1, \vec{v}_2, \ldots, \vec{v}_m$, where $m \in \{3, 4, \ldots\}$, are nonzero vectors, then
\begin{align}\label{eq:extv} \eta(\vec{v}_1,\vec{v}_m) \leq \eta(\vec{v}_1,\vec{v}_2)+ \eta(\vec{v}_2,\vec{v}_3), + \ldots, + \eta(\vec{v}_{m-1},\vec{v}_m). \end{align}
\end{lem}

\begin{proof}
The proof follows easily from Inequality~\ref{eq:st1}. 
\end{proof}

Now, we adopt the notation shown in Figure~\ref{fig:1-1} and formalize the proof outlined in Section~\ref{sapp:ideas}. We assume that the sub-curve on the right hand side of $\Omega$ in Figure~\ref{fig:sim} is $L_r$, and the other one is $L_l$, where we denote the set of ordered vertices of $L_r$ as 
$$\{v_0,v_1, \ldots, v_n\},$$ with $v_0=w$. 

We have $\theta(t_0) \leq \max_{t\in[0,1]} \theta(t)$. It is trivially true that $T_{\kappa}(L_r) \leq T_{\kappa}(L)$, so that with \textit{Condition 2:} $T_{\kappa}(L) + \max_{t\in[0,1]} \theta(t)< \frac{\pi}{2}$, we have
\begin{align}\label{eq:tp2t} T_{\kappa}(L_r) + \theta(t_0) \leq T_{\kappa}(L) + \max_{t\in[0,1]} \theta(t)< \frac{\pi}{2}.\end{align}

The statement and proof of Lemma~\ref{lem:tlpt} depend upon the point $w$ chosen in Step 1 of the Outline presented in Section~\ref{sapp:ideas}. There, the point $w$ was defined as an \emph{interior point} of a line segment $\hbar$ of $L$, so that $w$ is precluded from being a vertex of the original PL curve $L$.

\begin{lem}\label{lem:tlpt}
The plane $\Omega$ intersects $L$ only at the single point $w$. 
\end{lem}

\begin{proof}
Here we prove $\Omega \cap L_r=w$. A similar argument will show $\Omega \cap L_l=w$. 

The oriented initial line segment of $L_r$ is $\overrightarrow{wv_1}$ which lies on $\hbar$. So
$$\eta(\vec{n}_{\Omega}, \overrightarrow{wv_1})=  \eta(\vec{n}_{\Omega},\hbar) =\theta(t_0) < \frac{\pi}{2}.$$

For a proof by contradiction, assume that $\Omega$ intersects $L_r$ at some point $u$ other than $w$. The possibility that $\overrightarrow{wv_1} \subset \Omega$ is precluded by $\theta(t_0) < \pi/2$, so the plane $\Omega$ intersects $\overrightarrow{wv_1}$ only at $w$. So $u \notin \overrightarrow{wv_1}$.

\begin{figure}[h!]
\begin{center}
\includegraphics[height=5cm]{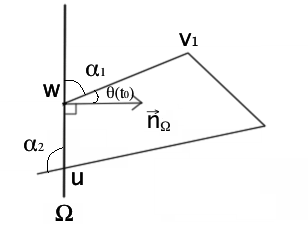}
\end{center}
\caption{The intersection $u$ generates a closed $PL$ curve}
\label{fig:cpl}
\end{figure}

Denote the sub-curve of $L_r$ from $w$ to $u$ as $L(wu)$. Then, since $u \notin \overrightarrow{wv_1}$, the union, $L(wu)\cup \overrightarrow{uw}$, forms a closed $PL$ curve, as Figure~\ref{fig:cpl} shows. By Fenchel's theorem we have 
\begin{align}\label{eq:twu2p} T_{\kappa}(L(wu) \cup\overrightarrow{uw}) \geq 2\pi.\end{align}
Denote the exterior angle of the $PL$ curve $L(wu)\cup \overrightarrow{uw}$ at $w$ as $\alpha_1$ (Figure~\ref{fig:cpl}), that is, 
$$\alpha_1=\eta(\overrightarrow{uw},\overrightarrow{wv_1}).$$ 
By Inequality~\ref{eq:st1},
$$\alpha_1 =\eta(\overrightarrow{uw},\overrightarrow{wv_1}) \leq  \eta(\overrightarrow{uw}, \vec{n}_{\Omega}) + \eta(\vec{n}_{\Omega},\overrightarrow{wv_1}).$$
Since $\overrightarrow{uw} \subset \Omega$, we have that $\eta(\overrightarrow{uw}, \vec{n}_{\Omega})=\frac{\pi}{2}$. Note also that $\eta(\vec{n}_{\Omega},\overrightarrow{wv_1})=\theta(t_0)$. So
$$\alpha_1 \leq \frac{\pi}{2} + \theta(t_0).$$
Denote the exterior angle of the $PL$ curve $L(wu)\cup \overrightarrow{uw}$ at $u$ as $\alpha_2$. By the definition of exterior angles, we have $\alpha_2 \leq \pi$, so that
$$T_{\kappa}(L(wu) \cup\overrightarrow{uw}) =\alpha_1+T_{\kappa}(L(wu))+\alpha_2$$ 
$$\leq \frac{\pi}{2} + \theta(t_0) + T_{\kappa}(L(wu)) + \pi.$$
It follows from Inequality~\ref{eq:twu2p} that
$$ \frac{\pi}{2} + \theta(t_0) + T_{\kappa}(L(wu)) + \pi \geq 2\pi,$$
so
\begin{align}\label{eq:tluv} T_{\kappa}(L(wu)) + \theta(t_0) \geq \frac{\pi}{2}.\end{align}
By $L(wu) \subset L_r$, we have
$$T_{\kappa}(L_r) +\theta(t_0) \geq T_{\kappa}(L(wu)) + \theta(t_0)  \geq \frac{\pi}{2}.$$
But this contradicts Inequality~\ref{eq:tp2t}. 
\end{proof}

\begin{lem}\label{lem:prol}
Any plane parallel to $\Omega$ intersects $L$ at no more than a single point. 
\end{lem}
\begin{proof}
Suppose $\tilde{\Omega}$ is a plane parallel to $\Omega$. If  $\tilde{\Omega} \cap L = \emptyset$, then we are done, so we assume that $\tilde{\Omega} \cap L \neq \emptyset$. If $\tilde{\Omega}=\Omega$, then Lemma~\ref{lem:tlpt} applies, so we also assume that $\tilde{\Omega} \neq \Omega$, implying that $w \notin \tilde{\Omega}$.

Consider two closed half-spaces $\mathbb{H}_l$ and $\mathbb{H}_r$ such that $\mathbb{H}_l \cup \mathbb{H}_r = \mathbb{R}^3$ and $\mathbb{H}_l \cap \mathbb{H}_r = \Omega$. Since $\Omega \cap L_l=\Omega \cap L_r=w$ and $L=L_l \cup L_r$ is simple, we can assume that $L_l \subset \mathbb{H}_l$ and $L_r \subset \mathbb{H}_r$. 

Suppose without loss of generality that $\tilde{\Omega} \subset H_r$, as shown in Figure~\ref{fig:para}.  Then since $L_l \subset \mathbb{H}_l$ and $\mathbb{H}_l \cap \mathbb{H}_r = \Omega \neq \tilde{\Omega}$, we have $\tilde{\Omega} \cap L_l =\emptyset$. Since we assumed $\tilde{\Omega} \cap L \neq \emptyset$, it follows that $\tilde{\Omega} \cap L_r \neq \emptyset$. Now, it suffices to show that $\tilde{\Omega} \cap L_r$ is a single point. 

Since $L_r$  is compact and oriented, let $\tilde{w}$ denote the first point of $L_r$, at which $\tilde{\Omega}$ intersects $L_r$. Since $\tilde{\Omega} \parallel \Omega$ and $\tilde{\Omega} \neq \Omega$, we have $\tilde{w} \neq w$. We shall show that $\tilde{\Omega} \cap L_r=\tilde{w}$.

\begin{figure}[h!]
\begin{center}
\includegraphics[height=5cm]{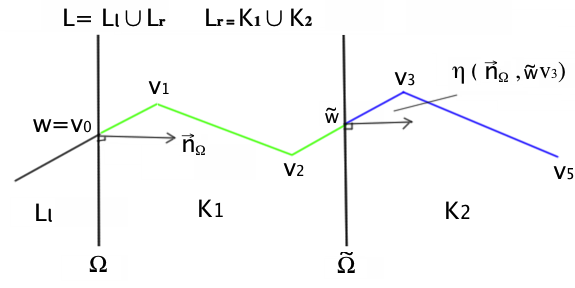}
\end{center}
\caption{A parallel plane intersecting $L$}
\label{fig:para}
\end{figure}

Denote the sub-curve of $L_r$ from its initial point $v_0$ to $\tilde{w}$ as $K_1$, and the sub-curve from $\tilde{w}$ to its end point $v_n$ as $K_2$, as shown in Figure~\ref{fig:para}. Since $\tilde{w}$ is the first intersection point of $\tilde{\Omega} \cap L_r$, but $K_1$ ends in $\tilde{w}$, then it is clear that $\tilde{\Omega} \cap K_1$ contains only $\tilde{w}$. Then in order to show $\tilde{\Omega} \cap L_r = \tilde{w}$, it suffices to show that $\tilde{\Omega} \cap K_2 = \tilde{w}$. 

If $\tilde{w}=v_n$, then it is the degenerate case: $K_2=\tilde{w}$, and we are done. Otherwise, there is a vertex $v_k$ for some $k\in \{1,\ldots,n\}$ such that $\overrightarrow{\tilde{w} v_k}$ is the non-degenerate initial segment of $K_2$, where $\tilde{w} \neq v_k$. Now we shall establish the inequality: 
$$T_{\kappa}(K_2) + \eta(\vec{n}_{\Omega}, \overrightarrow{\tilde{w} v_k}) < \frac{\pi}{2},$$ 
to guarantee a single point of intersection, similar to arguments previously given in Lemma~\ref{lem:tlpt}. To this end, we use Inequality~\ref{eq:st1} to note that
\begin{alignat}{2}\label{eq:exvwk}\eta(\vec{n}_{\Omega}, \overrightarrow{\tilde{w} v_k}) & \leq \eta(\vec{n}_{\Omega}, \overrightarrow{v_0v_1}) + \eta(\overrightarrow{v_0v_1}, \overrightarrow{\tilde{w} v_k})\\&=\theta(t_0)+\eta(\overrightarrow{v_0v_1}, \overrightarrow{\tilde{w} v_k}).\end{alignat}
The proof will be completed if we can show that
\begin{align}\label{eq:tk2}T_{\kappa}(K_2) + \theta(t_0)+\eta(\overrightarrow{v_0v_1}, \overrightarrow{\tilde{w} v_k}) < \frac{\pi}{2}.\end{align}

\textbf{ Case1:} The intersection $\tilde{w}$ is not a vertex, that is, $\tilde{w} \neq v_{k-1}$. Then $\tilde{w}$ is an interior point of $\overrightarrow{v_{k-1}v_k}$, and hence $T_{\kappa}(K_1)=\eta(\overrightarrow{v_0v_1}, \overrightarrow{v_1v_2})+\ldots+\eta(\overrightarrow{v_{k-2}v_{k-1}}, \overrightarrow{v_{k-1}\tilde{w}})$, and $\eta(\overrightarrow{v_{k-1}\tilde{w}}, \overrightarrow{\tilde{w} v_k})=0$. By Lemma~\ref{lem:bt},  
$$\eta(\overrightarrow{v_0v_1}, \overrightarrow{\tilde{w} v_k})$$
$$\leq \eta(\overrightarrow{v_0v_1}, \overrightarrow{v_1v_2})+\ldots+\eta(\overrightarrow{v_{k-2}v_{k-1}}, \overrightarrow{v_{k-1}\tilde{w}}) + \eta(\overrightarrow{v_{k-1}\tilde{w}}, \overrightarrow{\tilde{w} v_k})$$ 
$$= T_{\kappa}(K_1).$$
So 
$$T_{\kappa}(K_2) + \theta(t_0)+\eta(\overrightarrow{v_0v_1}, \overrightarrow{\tilde{w} v_k})  \leq T_{\kappa}(K_2) + \theta(t_0) + T_{\kappa}(K_1).$$
We also have
$$T_{\kappa}(L_r)=T_{\kappa}(K_1)+ \eta(\overrightarrow{v_{k-1}\tilde{w}}, \overrightarrow{\tilde{w} v_k})+T_{\kappa}(K_2)=T_{\kappa}(K_1)+T_{\kappa}(K_2),$$
(since $\eta(\overrightarrow{v_{k-1}\tilde{w}}, \overrightarrow{\tilde{w} v_k})=0$), so that
$$T_{\kappa}(K_2) + \theta(t_0)+\eta(\overrightarrow{v_0v_1}, \overrightarrow{\tilde{w} v_k})  \leq T_{\kappa}(L_r) +\theta(t_0),$$
which is less than $\frac{\pi}{2}$, by Inequality~\ref{eq:tp2t}.

\vspace{1ex}

\textbf{ Case2:} The intersection $\tilde{w}$ is a vertex, that is, $\tilde{w}=v_{k-1}$, then $T_{\kappa}(K_1)=\eta(\overrightarrow{v_0v_1}, \overrightarrow{v_1v_2})+\ldots+\eta(\overrightarrow{v_{k-3}v_{k-2}}, \overrightarrow{v_{k-2}\tilde{w}})$. By Lemma~\ref{lem:bt},  
\begin{alignat}{3}
& \ \ \ \  \eta(\overrightarrow{v_0v_1}, \overrightarrow{\tilde{w} v_k}) \nonumber \\
& \leq \eta(\overrightarrow{v_0v_1}, \overrightarrow{v_1v_2})+\ldots+\eta(\overrightarrow{v_{k-3}v_{k-2}}, \overrightarrow{v_{k-2}\tilde{w}})+\eta(\overrightarrow{v_{k-2}\tilde{w}}, \overrightarrow{\tilde{w} v_k}) \nonumber \\
& \leq T_{\kappa}(K_1)+ \eta(\overrightarrow{v_{k-2}\tilde{w}}, \overrightarrow{\tilde{w} v_k}). \nonumber
\end{alignat}
So
$$T_{\kappa}(K_2) + \theta(t_0)+\eta(\overrightarrow{v_0v_1}, \overrightarrow{\tilde{w} v_k})$$ 
$$\leq T_{\kappa}(K_2) + \theta(t_0)+ T_{\kappa}(K_1)+ \eta(\overrightarrow{v_{k-2}\tilde{w}}, \overrightarrow{\tilde{w} v_k}).$$
But by the definition of the total curvature for a $PL$ curve,
$$T_{\kappa}(K_2) + T_{\kappa}(K_1)+ \eta(\overrightarrow{v_{k-2}\tilde{w}}, \overrightarrow{\tilde{w} v_k})=T_{\kappa}(L_r).$$
So 
$$T_{\kappa}(K_2) + \theta(t_0)+\eta(\overrightarrow{v_0v_1}, \overrightarrow{\tilde{w} v_k}) \leq T_{\kappa}(L_r) +\theta(t_0),$$
which is less than $\frac{\pi}{2}$, by Inequality~\ref{eq:tp2t}.

So Inequality~\ref{eq:tk2} holds, which is an inequality analogous to Inequality~\ref{eq:tp2t}.  If in the proof of Lemma~\ref{lem:tlpt}, we change $\Omega$ to $\tilde{\Omega}$, $L_r$ to $K_2$ and $\theta(t_0)$ to $\eta(\vec{n}_{\Omega}, \overrightarrow{\tilde{w} v_k})$, then a similar proof of Lemma~\ref{lem:tlpt} will show that $\tilde{\Omega} \cap K_2 = \tilde{w}$. This completes the proof. 
\end{proof}

\begin{lem}\label{lem:dl1}
For an arbitrary $t_0 \in [0,1]$, the disc $D_r(t_0)$ intersects $C$ at a unique point, and also intersects $L$ at a unique point.
\end{lem}

\begin{proof}
First, we have $C(t_0) \in D_r(t_0) \cap C$. If there is an additional point, say $C(t_1 )\in D_r(t_0) \cap C$ where $t_1  \neq t_0$, then we have that $C(t_1) \neq C(t_0)$ because $C$ is simple, and hence $D(t_1) \neq D(t_0)$. Since also $C(t_1 )\in D_r(t_1)$, we have that $C(t_1) \in D_r(t_0) \cap D_r(t_1)$. But this contradicts the non-self-intersection of $\Gamma$. So $D_r(t_0) \cap C$ must be a unique point.

Now, we show that $D_r(t_0) \cap L \neq \emptyset$. If $t_0=0$ or $t_0=1$, then since $L(0) \in D_r(0)$ and $L(1) \in D_r(1)$, we have that  $D_r(t_0) \cap L \neq \emptyset$. 

Otherwise if $t_0 \in (0,1)$, then assume to the contrary that $D_r(t_0) \cap L =\emptyset$. Since $L \subset \Gamma$ by \textit{Condition 1}, the contrary assumption implies that $L \subset \Gamma \setminus D_r(t_0)$. Because $C$ is an open curve, we have that $D_r(0) \neq D_r(1)$. So  $\Gamma \setminus D_r(t_0)$ consists of two disconnected components, but this implies that $L$ is disconnected, which is a contradiction. So 
\begin{align}\label{eq:Drl}D_r(t_0) \cap L \neq \emptyset.\end{align}

Since $D_r(t_0) \parallel \Omega$ (as discussed in Section~\ref{sapp:ideas}), Lemma~\ref{lem:prol} implies that the plane containing $D_r(t_0)$ intersects $L$ at no more than a single point, which, of course, further implies that $D_r(t_0)$ intersects $L$ at no more than a single point. This plus Inequality~\ref{eq:Drl} shows that $D_r(t_0) \cap L$ is a single point. 

If $D_r(t_0) \cap L = D_r(t_1) \cap L$ for some $t_1 \neq t_0$, then $D_r(t_0)$ and $D_r(t_1)$ intersect, which contradicts the non-self-intersection of $\Gamma$. So there is an one-to-one correspondence between the parameter $t$ and  the point $D_r(t) \cap L$ for $t\in[0,1]$, which shows the uniqueness. 
\end{proof}
\begin{lem}\label{lem:lin}
The map $\tilde{L}(t)$ given by Equation~\ref{eq:tl} is well defined, one-to-one and onto.
\end{lem}

\begin{proof}
It is well defined by Lemma~\ref{lem:dl1}. Suppose $\tilde{L}(t_1)=\tilde{L}(t_2)$, then $D_r(t_1) \cap L=D_r(t_2) \cap L$ which is not empty by Lemma~\ref{lem:dl1}. So $D_r(t_1) \cap D_r(t_2) \neq \emptyset$. Since $\Gamma$ is nonsingular, it follows that $D_r(t_1) = D_r(t_2)$. Since $C$ is simple, if $D_r(t_1)=D_r(t_2)$, then $t_1=t_2$. Thus $\tilde{L}$ is one-to-one. Since $L \subset \Gamma$, each point of $L$ is contained in some disc $D_r(t)$. So $\tilde{L}$ is onto. 
\end{proof}

\begin{lem}\label{lem:contl}
The map $\tilde{L}(t) $ given by Equation~\ref{eq:tl} is continuous. 
\end{lem}

\begin{proof}
Let $\Gamma_{t_1t_2}$ be the portion of $\Gamma$ corresponding to $[t_1,t_2]$, that is
$$\Gamma_{t_1t_2}=\bigcup_{t\in[t_1,t_2]}D_r(t).$$
Suppose that $s \in [0,1]$ is an arbitrary parameter. Then by Lemma~\ref{lem:lin}, there is a unique point $q \in L$ such that $q=\tilde{L}(s)=D_r(s) \cap L$. We shall prove the continuity of $\tilde{L}(t)$ at $s$ by the definition, that is, for $\forall \epsilon>0$, there exists a $\delta>0$ such that $|t-s|<\delta$ implies $||\tilde{L}(t) -\tilde{L}(s)||<\epsilon$. 

Note that $D_r(s)$ divides $\Gamma$ into $\Gamma_{0s}$ and $\Gamma_{s1}$. Since $C$ is an open curve, it follows that $D_r(0) \neq D_r(1)$, and that $\Gamma_{0s}$ and $\Gamma_{s1}$ intersect at only $D_r(s)$.  By Lemma~\ref{lem:dl1}, $D_r(s) \cap L$ is a single point, so $L$ is divided by $D_r(s)$ into two sub-curves, denoted as $K_1$ and $K_2$, that is  $K_1 \subset \Gamma_{0s}$ and $K_2 \subset \Gamma_{s1}$, as shown in Figure~\ref{fig:cont}. 

\begin{figure}[h!]
\begin{center}
\includegraphics[height=5cm]{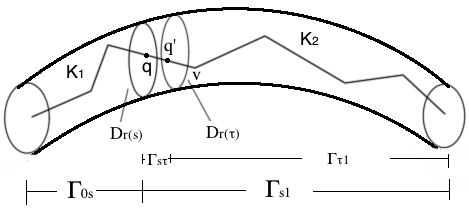}
\end{center}
\caption{If $|s-\tau|<\delta$, then $||q-q'|| <\epsilon$}
\label{fig:cont}
\end{figure}

\textbf{ Case1:}  The parameter $s$ is such that $s \neq 0$ and $s \neq 1$. Consider $\Gamma_{s1}$ first. Since $K_2$ is oriented, we can let $v$ be the first vertex of $K_2$ that is nearest (in distance along $K_2$) to $q$. For any $0<\epsilon<||\overline{qv}||$, let $q' \in \overline{qv}$ such that $||\overline{qq'}||=\epsilon$. By Lemma~\ref{lem:lin}, $q'=\tilde{L}(\tau)=D_r(\tau) \cap L$ for some $\tau \in (s,1]$. 

First, note $\overline{qq'} \cap \textup{int}\Gamma_{s\tau} \neq \emptyset$. To verify this, observe $\overline{qq'} \subset \overline{qv} \subset K_2 \subset \Gamma_{s1}$ and $\Gamma_{s1}=\Gamma_{s\tau} \cup \Gamma_{\tau 1}$, so $\overline{qq'} \subset \Gamma_{s\tau} \cup \Gamma_{\tau 1}$. If $\overline{qq'} \cap \textup{int}\Gamma_{s\tau} = \emptyset$, then the segment $\overline{qq'}$ is contained in $D_r(s) \cup \Gamma_{\tau1}$ which is disconnected. This implies $\overline{qq'}$ is disconnected, which is a contradiction. 

Secondly, note that the subset $\Gamma_{s\tau}$ of a nonsingular pipe section is connected (since $C$ is $C^1$), and $\overline{qq'}$ is a line segment jointing the end discs of $\Gamma_{s\tau}$, and has intersections with interior of $\Gamma_{s\tau}$. This geometry implies that 
\begin{align}\label{eq:qq'}\overline{qq'} \subset \Gamma_{s \tau}.\end{align}

Let $\delta=\tau-s$. For an arbitrary $t \in (s,s+\delta)=(s,\tau)$,  Inclusion~\ref{eq:qq'} implies that $\tilde{L}(t)=D_r(t) \cap \overline{qq'}$. Since neither $\tilde{L}(t) \neq q$ or $\tilde{L}(t) \neq q'$, it follows that $\tilde{L}(t) \in \textup{int}(\overline{qq'})$. So 
$$||\tilde{L}(t)-\tilde{L}(s)||<||\overline{qq'}||=\epsilon.$$
This shows the right-continuity. We similarly consider the $\Gamma_{0s}$ and obtain the left-continuity.

\vspace{1ex}

\textbf{ Case2:} The parameter $s$ is such that $s=0$ or $s=1$. We similarly obtain the right-continuity if $s=0$, or the left-continuity if $s=1$. 
\end{proof}

\begin{theo}\label{thm:hhom}
If $L$ and $C$ satisfy \textit{Conditions 1 \textup{and} 2}, then the map $h$ defined by Equation~\ref{eq:hh} is a homeomorphism.
\end{theo}

\begin{proof}
By Lemma~\ref{lem:lin}, $\tilde{L}(t)$ is one-to-one and onto. By Lemma~\ref{lem:contl}, $\tilde{L}(t)$ is continuous. Since $\tilde{L}$ is defined on a compact domain, it is a homeomorphism.

Note that $C$ is simple and open, so $C(t)$ is one-to-one, and it is obviously onto. The map $C(t)$ is also continuous and defined on a compact domain, so $C(t)$ a homeomorphism. Since $h$ is a composition of $C^{-1}$ and $\tilde{L}$, $h$ is a homeomorphism.  
\end{proof} 

\begin{rem}
A very natural way to define a homeomorphism between simple curves $C$ and $L$ would be by $f(p) = L(C^{-1}(p))$.  An easy method to extend $f$ to a homotopy is the straight-line homotopy.  However, we were not able to establish that a straight-line homotopy based upon $f$ would also be an isotopy, where it would be necessary to show that each pair of line segments generated is disjoint. Our definition of $h$  in Equation~\ref{eq:hh} was strategically chosen so that this isotopy criterion is easily established, since the normal discs are already pairwise disjoint.
\end{rem}

\section{Construction of Ambient Isotopies}\label{ssec:cab}
Note that $L$ and $C$ fit inside a nonsingular pipe section $\Gamma$ of $C$. For a similar problem, an explicit construction has appeared \cite[Section 4.4]{Lance2009} \cite{TJP2011}. The proof of Lemma~\ref{lem:eamlc}, below, is a simpler version of a previous proof  \cite[Corollary 4]{TJP2011}. The construction here relies upon some basic properties of convex sets, which are repeated here.  For clarity, the complete proof of Lemma~\ref{lem:eamlc} is given here.

\begin{figure}[h!]
\centering
    \subfigure[Rays outward]
    {
   \includegraphics[height=3.5cm]{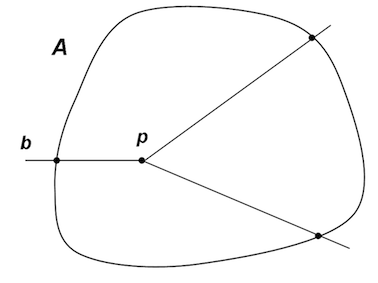} 
\label{fig:cvx}
    }
    \subfigure[Variant of a push]
    {
    \includegraphics[height=3.5cm]{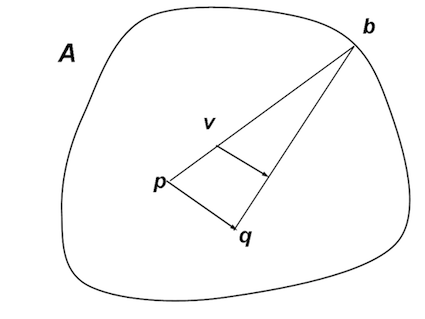}
\label{fig:cpush}
    }
    \caption{Convex subset}
\end{figure}

The Images in Figures~\ref{fig:cvx} and~\ref{fig:cpush} were created by L. E. Miller  and are used, here, with permission.

\begin{lem}\label{lem:asf}\cite[Lemma 6]{TJP2011}
Let $A$ be a compact convex subset of $\mathbb{R}^2$ with non-empty interior. For each point $p \in int(A)$ and $b \in \partial A$, the ray going from $p$ to $b$
only intersects $\partial A$ at $b$ (See Figure~\ref{fig:cvx}.)
\end{lem}

\begin{lem}\label{lem:lapb}\cite[Lemma 7]{TJP2011}
Let $A$ be a compact convex subset of $\mathbb{R}^2$ with non-empty interior and fix $p \in int(A)$. For each boundary point $b \in \partial A$, denote by $[p, b]$ the line segment from $p$ to $b$. Then $A=\bigcup_{b \in \partial A} [p, b]$.
\end{lem}

\begin{lem}\label{lem:eamlc}
There is an ambient isotropy between $L$ and $C$ with compact support of $\Gamma$, leaving $\partial \Gamma$ fixed.
\end{lem}

\begin{proof}
We consider each normal disc $D_r(t)$ for $t\in[0,1]$. Let $p=D_r(t) \cap C$ and $q=h(p)$ with $h$ defined by Equation~\ref{eq:hh}, then define a map $F_{p,q}: D_r(t) \rightarrow D_r(t)$ such that it sends each line segment $[p,b]$ for $b \in \partial D_r(t)$, linearly onto the line segment $[q,b]$ as Figure~\ref{fig:cpush} shows. The two previous lemmas (Lemma~\ref{lem:asf} and Lemma~\ref{lem:lapb}), will yield that $F_{p,q}$ is a homeomorphism, leaving $\partial D_r(t)$ fixed \cite[Lemma 4.4.6]{Lance2009}. 

In order to extend $F_{p,q}$ to an ambient isotopy, define $H: D_r(t) \times [0,1] \rightarrow D_r(t)$  \cite[Corollary 4.4.7]{Lance2009} by
\[
H(v,s)=\left\{
\begin{array}{l l}
  (1-s)p+sq  & \quad \text{if $v=p$}\\
  F_{p,(1-s)p+sq}(v)  &   \quad \text{if $v \neq p$},\\
\end{array} \right.
\]
where $F_{p,(1-s)p+sq}$ is a map on $D_r(t)$ analogous to $F_{p,q}$, sending each line segment $[p,b]$ for $b \in \partial D_r(t)$, linearly onto the line segment $[(1-s)p+sq,b]$. 


It is a routine \cite[Corollary 4.4.7]{Lance2009} to verify that $H(v,s)$ is well defined on the compact set $D_r(t)$, continuous, one-to-one and onto, leaving $\partial D_r(t)$ fixed. Now, we naturally define an ambient isotopy $T_t: \mathbb{R}^2 \times [0,1] \rightarrow \mathbb{R}^2$ on the plane containing $D_r(t)$ by
\[
T_t(v,s)=\left\{
\begin{array}{l l}
  H(v,s)  & \quad \text{if $v \in D_r(t)$}\\
  v  &   \quad \text{otherwise}.\\
\end{array} \right.
\]

We then define $T: \mathbb{R}^3 \times [0,1] \rightarrow \mathbb{R}^3$  by 
\[
T(v,s)=\left\{
\begin{array}{l l}
  T_t(v,s)  & \quad \text{if $v\in D_r(t)$}\\
  v  &   \quad \text{otherwise}.\\
\end{array} \right.
\]
The fact that the normal discs $D_r(t)$ are disjoint ensures that $T$ is an ambient isotopy \cite[Corollary 4.4.8]{Lance2009}, with compact support of $\Gamma$, leaving $\partial \Gamma$ fixed. 


\end{proof}

\section{Ambient Isotopy for B\'ezier Curves}
Now we apply this result to a simple, regular, composite, $C^1$ B\'ezier curve $\mathcal{B}$ and the control polygon $\mathcal{P}$.

\subsection{Ambient isotopy}
There exists a nonsingular pipe surface~\cite{Maekawa_Patrikalakis_Sakkalis_Yu1998} of radius $r$ for $\mathcal{B}$, denoted as  $S_{\mathcal{B}}(r)$. Denote the nonsingular pipe section determined by $S_{\mathcal{B}}(r)$ as $\Gamma_{\mathcal{B}}$. Also, for each sub-control polygon of $\mathcal{B}$, there exists a corresponding nonsingular pipe sections. Denote the nonsingular pipe section corresponding to the $k$th control polygon as $\Gamma_k$. 

\begin{theo}\label{thm:eampb}
Each sub-control polygon $P^k$ of a B\'ezier curve $\mathcal{B}$ will eventually satisfy \textit{Conditions 1 \textup{and} 2} via subdivision, and consequently, there will be an ambient isotropy between $\mathcal{B}$ and $\mathcal{P}$ with compact support of $\Gamma_{\mathcal{B}}$, leaving $\partial \Gamma_{\mathcal{B}}$ fixed.
\end{theo}

\begin{proof}
By Lemma~\ref{lem:hp12}, \textit{Conditions 1 \textup{and} 2} can be achieved by subdivisions. Now consider each sub-control polygon $P^k$ satisfying \textit{Conditions 1 \textup{and} 2}, and the corresponding B\'ezier sub-curves $B^k$. Use Lemma~\ref{lem:eamlc} to define an ambient isotopy $\Psi_k: \mathbb{R}^3 \times [0,1] \rightarrow \mathbb{R}^3$ between $B^k$ and $P^k$, for each $k \in \{ 1,2,\ldots,2^i\}$. Define $\Psi:  \mathbb{R}^3 \times [0,1] \rightarrow \mathbb{R}^3$ by the composition
$$\Psi=\Psi_1 \circ \Psi_2 \circ \ldots \circ \Psi_{2^i}. $$
Note that $\Psi_k$ fixes the complement of $\textup{int}(\Gamma_k)$, and $\textup{int}(\Gamma_k) \cap \textup{int}(\Gamma_{k'}) = \emptyset $ for all $k \neq k'$. So the composition $\Psi$ is well defined. Since each $\Psi_k$ is an ambient isotopy, the composition $\Psi$ is an ambient isotopy between $\mathcal{B}$ and $\mathcal{P}$ with compact support of $\Gamma_{\mathcal{B}}$, leaving $\partial \Gamma_{\mathcal{B}}$ fixed.  
\end{proof}
\subsection{Sufficient subdivision iterations}\label{ssec:nchns}
Now we consider sufficient numbers of subdivision iterations to achieve the ambient isotopy defined by Theorem~\ref{thm:eampb}, that is, we shall have a control polygon that satisfies \textit{Conditions 1 \textup{and} 2}. The number of subdivisions for \textit{Condition 1} is given in the paper \cite[Lemma 6.3]{bez-iso}. To obtain the number of subdivisions for $\textit{Condition 2}$, we consider the following, for which we let $\mathcal{P'}(t)=l'(P,i)(t)$ (the first derivative of the control polygon $\mathcal{P}$), and denote the angle between $\mathcal{B'}(t)$ and $\mathcal{P'}(t)$ as $\theta(t)$, for $t\in[0,1]$. 

\begin{lem}\textup{\cite[Theorem 6.1]{bez-iso}}\label{lem:tmt1}
For any $0< \nu < \frac{\pi}{2}$, there is an integer $N(\nu)$ such that each exterior angle of $\mathcal{P}$ will be less than $\nu$ after $N(\nu)$ subdivisions, where
\begin{align}\label{eq:N} N(\nu)=\lceil \max\{N_1,\log(f(\nu))\} \rceil, \end{align}
$$N_1= \frac{1}{2} \log(\frac{N_{\infty}(n-1)\parallel \triangle_2P' \parallel}{\sigma}),$$ and
$$ f(\nu)=\frac{2M}{(1-\cos(\nu))(\sigma-B'_{dist}(N_1))}. $$
\end{lem}

Note that for a B\'ezier curve of degree $n$, there are $n-1$ exterior angles for each sub-control polygon $P^k$. To have $T_{\kappa}(P^k) <\frac{\pi}{2}$, it suffices to make each exterior angle smaller than $\frac{\pi}{2(n-1)}$. By Lemma~\ref{lem:tmt1}, this can be gained by $N(\frac{\pi}{2(n-1)})$ subdivisions. 

\textit{Condition 2} is motivated by the weaker condition $T_{\kappa}(P^k) <\frac{\pi}{2}$. We couldn't derive the same results by using this weaker condition instead, but our \textit{Condition 2} requires at most one more subdivision, as shown below. 

\begin{lem}\label{lem:c2wf}
\textit{Condition 2} will be fulfilled by at most $N(\frac{\pi}{2(n-1)})+1$ subdivisions.
\end{lem}
\begin{proof}
To prove $$T_{\kappa}(P^k) + \max_{t\in[0,1]} \theta(t) <\frac{\pi}{2},$$
it suffices to prove  
$$T_{\kappa}(P^k) < \frac{\pi}{6}\ \text{and}\ \max_{t\in[0,1]} \theta(t) <\frac{\pi}{3}.$$
To have $T_{\kappa}(P^k) < \frac{\pi}{6}$, by Lemma~\ref{lem:tmt1}, $N(\frac{\pi}{6(n-1)})$ subdivisions will be sufficient. The definition given by Equation~\ref{eq:N} implies that 
$$N(\frac{\pi}{6(n-1)}) \leq N(\frac{\pi}{2(n-1)}) +1.$$
On the other hand, by \cite[Section 6.3]{bez-iso}, for all $t\in [0,1]$, we have 
$$1- \cos (\theta(t))  \leq \frac{2B'_{dist}(i)}{\sigma},$$
where 
$$B'_{dist}(i) := \frac{1}{2^{2i}}N_{\infty}(n-1) || \Delta_2P' ||.$$
So to have $\max_{t\in[0,1]} \theta(t) <\frac{\pi}{3}$, it suffices to set
$$\frac{2B'_{dist}(i)}{\sigma} < \frac{\pi}{3},$$
which implies 
$$ i \geq \frac{1}{2}\log (\frac{N_{\infty}(n-1) ||\Delta_2 P'||}{\sigma})+1.$$
Comparing it with Equation~\ref{eq:N} , we find that it is at most one more than $N(\nu)$ for any $0< \nu < \frac{\pi}{2}$. The conclusion follows. 
\end{proof}

Let 
\begin{align}\label{eq:nsta}  N^{\star}=\max\{N(\frac{\pi}{2(n-1)})+1 , N'(r)\}, \end{align}
where $r$ is the radius of $S_r(\mathcal{B})$. 

\begin{theo}\label{thm:pnm}
Performing $N^{\star}$ or more subdivisions, where $N^{\star}$ is given by Equation~\ref{eq:nsta}, will produce an ambient isotopic $\mathcal{P}$ for $\mathcal{B}$.
\end{theo}

\begin{proof}
According to \cite[Lemma 6.3]{bez-iso}, \textit{Condition 1} is satisfied after $N'(r)$ subdivisions. By Lemma~\ref{lem:c2wf}, \textit{Condition 2} is satisfied after $N(\frac{\pi}{2(n+1)})+1$ subdivisions. Then Theorem~\ref{thm:eampb} can be applied to draw the conclusion. 
\end{proof}

Now we compare this result with the existing one \cite{bez-iso}.

\begin{rem}
\label{re:better}
To obtain ambient isotopy, the previously established result \textup{\cite{bez-iso}} needs $\max\{N(\frac{\pi}{2(n-1)}) , N'(r)\} +2$ subdivision iterations \textup{\cite[Remark 6.1]{bez-iso}}. In contrast, Theorem~\ref{thm:pnm} implies $\max\{N(\frac{\pi}{2(n-1)}) , N'(r)\} +1$ will be sufficient. A subdivision doubles the number of line segments. Therefore, with only one less subdivision, the work here produces much less line segments, which may be useful especially for applications with very complex shapes.  
\end{rem}
\section{Conclusions}

We give two conditions regarding distance, and total curvature combined with derivative, to guarantee the same knot type. It can be directly applied to B\'ezier curves. This work is alternative to an existence result of requiring the containment of convex hulls of sub-control polygons, and another result using conditions of distance and derivative. The approach here allows fewer subdivision iterations and less line segments by explicitly constructing ambient isotopies. Moreover, we showed that it is possible to verify the condition of total curvature only, other than total curvature combined with derivative, with a price of one additional subdivision. Testing the global property of total curvature may be easier than testing the local property of derivative in some practical situations. It may find applications in computer graphics, computer animation and scientific visualization. 

\bibliographystyle{plain}
\bibliography{ji-tjp-biblio}

\end{document}